\documentclass[a4paper]{jpconf}
\usepackage{latexsym,amsmath,amscd,amssymb,graphics,color}
\usepackage{amsthm,url,bm,enumerate}
\usepackage{graphicx}
\usepackage{epsfig, picinpar,framed}
\usepackage[margin=2cm]{geometry}
\newcommand{\bfi}{\bfseries\itshape}

\newcommand{\rem}[1]{}
\def\b{\begin{eqnarray}}
\def\e{\end{eqnarray}}

\newtheorem{theorem}{Theorem}%[section]

\newtheorem{lemma}[theorem]{Lemma}
\newtheorem{remark}[theorem]{Remark}

\begin{document}

%%%%%%%%%%%%%%%%%%%%%%%%%%%%%%%%%
%%%%%%%%

\title{ Euler-Poincar\'e equations for
$G$-Strands}

\author{ Darryl D. Holm$^a$,  Rossen I. Ivanov$^b$}

\address{ $^a$ Department of
Mathematics, Imperial College, London SW7 2AZ, UK \\
$^b$Department of Mathematical Sciences, Dublin Institute of
Technology, Kevin Street, Dublin 8, Ireland}

\ead{d.holm@imperial.ac.uk, rossen.ivanov@dit.ie }

\rem{
PACS numbers:
}

\begin{abstract}

The $G$-strand equations for a map $\mathbb{R}\times \mathbb{R}$ into a Lie group $G$ 
are associated to a $G$-invariant Lagrangian. The Lie group manifold is also the
configuration space for the Lagrangian. The $G$-strand itself is the
map $g(t,s): \mathbb{R}\times \mathbb{R}\to G$, where
$t$ and $s$ are the independent variables of the $G$-strand
equations. The Euler-Poincar\'e reduction of the variational
principle leads to a formulation where the dependent variables of
the $G$-strand equations take values in the corresponding Lie
algebra $\mathfrak{g}$ and its co-algebra, $\mathfrak{g}^*$ with
respect to the pairing provided by the variational derivatives of the Lagrangian.

We review examples of different $G$-strand constructions,
including matrix Lie groups and diffeomorphism group.
In some cases the $G$-strand equations are completely integrable
1+1 Hamiltonian systems that admit soliton solutions.

\end{abstract}

\section{Introduction}

We give a brief account of the $G$-strand construction, which
gives rise to equations for a map $\mathbb{R}\times \mathbb{R}$ into a Lie group $G$ 
associated to a $G$-invariant Lagrangian. Our presentation reviews our
previous works \cite{ Ho-Iv-Pe,Ho-Iv1, Ho-Iv2, FDT, HoLu2013} and
is aimed to illustrate the $G$-strand construction with several simple but
instructive examples. The following examples are reviewed here:

(i) $SO(3)$-strand equations for the so-called continuous spin
chain. The equations reduce to the integrable chiral model in
their simplest (bi-invariant) case.

(ii) $SO(3)$ - anisotropic chiral model, which is also integrable,

(iii) ${\rm Diff}(\mathbb{R})$-strand equations. These equations
are in general non-integrable; however they admit solutions in
$2+1$ space-time with singular support (e.g., peakons).
Peakon-antipeakon collisions governed by the 
${\rm Diff}(\mathbb{R})$-strand equations can be solved \emph{analytically},
and potentially they can be applied in the theory of image
registration.

\section{Ingredients of Euler--Poincar\'e theory for Left $G$-Invariant Lagrangians}

 Let $G$ be a Lie group. A map $g(t,s): \mathbb{R}\times
\mathbb{R}\to G$ has two types of tangent vectors, $\dot{g} := g_t
\in T G$ and $g' :=g_s \in T G$. Assume that the Lagrangian
density function $ L(g,\dot{g},g') $ is left $G$-invariant. The
left $G$--invariance of $L$ permits us to define $l:
\mathfrak{g}\times \mathfrak{g} \rightarrow \mathbb{R}$ by
\[
L(g,\dot{g},g')=L(g^{-1}g,g^{-1}\dot{g},g^{-1}g')\equiv l(g^{-1}
\dot{g} , g^{-1} g' ).
\]
Conversely,  this relation defines for any reduced lagrangian
$l=l({\sf u},{\sf v}) : \mathfrak{g}\times \mathfrak{g}
\rightarrow \mathbb{R} $ a left $G$-invariant function $ L : T
G\times TG \rightarrow \mathbb{R} $ and a map $g(t,s):
\mathbb{R}\times \mathbb{R}\to G$ such that
\[
{\sf u} (t,s) := g^{ -1} g_t (t,s) =g^{ -1}\dot{g}(t,s)
\quad\hbox{and}\quad {\sf v} (t,s) := g^{ -1} g_s (t,s)= g^{ -1}
g' (t,s) .\]

\begin{lemma}
The left-invariant tangent vectors ${\sf u} (t,s)$ and ${\sf v}
(t,s)$ at the identity of $G$ satisfy
\begin{equation}
{\sf v}_t - {\sf u}_s = -\,{\rm ad}_{\sf u}{\sf v} \,.
\label{zero-curv1}
\end{equation}
\end{lemma}

\begin{proof}
The proof is standard and follows from equality of cross
derivatives $g_{ts}=g_{st}$.

Equation (\ref{zero-curv1}) is usually called a {\bfi
zero-curvature relation}.
\end{proof}

\begin{theorem} [ Euler-Poincar\'e theorem for left-invariant Lagrangians]\label{lall}$\,$

With the preceding notation, the following two statements are
equivalent:
\begin{enumerate}
\item [{\bf i} ] Variational principle on $T G\times TG$ $\,\,$ $
\delta \int _{t_1} ^{t_2} L(g(t,s), \dot{g} (t,s), g'(t,s) )
\,ds\,dt = 0 $ holds, for variations $\delta g(t,s)$ of $ g (t,s)
$ vanishing at the endpoints in $t$ and $s$. The function $g(t,s)$
satisfies Euler--Lagrange equations for $L$ on $G$, given by
\begin{equation*} \label{EL-eqns}
\frac{\partial L}{\partial g} - \frac{\partial}{\partial
t}\frac{\partial L}{\partial g_t} - \frac{\partial}{\partial
s}\frac{\partial L}{\partial g_s} = 0.
\end{equation*}

\item [{\bf ii} ]  The constrained variational principle%
\footnote{As with the basic Euler--Poincar\'e equations, this is
not strictly a variational principle in the same sense as the
standard Hamilton's principle. It is more like the Lagrange
d'Alembert principle, because we impose the stated constraints on
the variations allowed.}
\begin{equation*} \label{variationalprinciple}
\delta \int _{t_1} ^{t_2}  l({\sf u}(t,s), {\sf v}(t,s)) \,ds\,dt
= 0
\end{equation*}
holds on $\mathfrak{g}\times\mathfrak{g}$, using variations of $
{\sf u} := g^{ -1} g_t (t,s)$ and ${\sf v}:= g^{ -1} g_s(t,s) $ of
the forms
\begin{equation*} \label{epvariations}
\delta {\sf u} = \dot{{\sf w} } + {\rm ad}_{\sf u}{\sf w}
\quad\hbox{and}\quad \delta {\sf v} = {\sf w}\,' + {\rm ad}_{\sf
v} {\sf w} \,,
\end{equation*}
where ${\sf w}(t,s) :=g^{ -1}\delta g \in \mathfrak{g}$ vanishes
at the endpoints. The {\bfi Euler--Poincar\'{e}} equations hold on
$\mathfrak{g}^*\times\mathfrak{g}^*$ ({\bfi$G$-strand equations})

\begin{align*}
\frac{d}{dt} \frac{\delta l}{\delta {\sf u}} -
 \operatorname{ad}_{{\sf u}}^{\ast} \frac{ \delta l }{ \delta {\sf u}}
+ \frac{d}{ds} \frac{\delta l}{\delta {\sf v}} -
 \operatorname{ad}_{{\sf v}}^{\ast} \frac{ \delta l }{ \delta {\sf v}}
 =
 0  \quad\hbox{ \& }\quad
\partial_{s}{\sf u} - \partial_t{\sf v} = [\,{\sf u},\,{\sf v}\,] = {\rm ad}_{\sf u}{\sf v}
\label{GSeqns}
\end{align*}
where $({\rm ad}^*:
\mathfrak{g}\times\mathfrak{g}^*\to\mathfrak{g}^*)$ is defined via
$({\rm ad}:\mathfrak{g}\times\mathfrak{g}\to\mathfrak{g})$ in the
dual pairing $\langle \,\cdot\,,\,\cdot\,\rangle:
\mathfrak{g}^*\times\mathfrak{g}\to\mathbb{R}$ by,
\bigskip
\begin{align*}
\left\langle {\rm ad}^*_{\sf u}\frac{\delta \ell}{\delta{\sf u}}
\,,\, {\sf v} \right\rangle_\mathfrak{g} = \left\langle
\frac{\delta \ell}{\delta{\sf u}} \,,\, {\rm ad}_{\sf u}{\sf v}
\right\rangle_\mathfrak{g}.
\end{align*}

\end{enumerate}
\end{theorem}

In 1901 Poincar\'e in his famous work proves that, when a Lie
algebra acts locally transitively on the configuration space of a
Lagrangian mechanical system, the well known Euler-Lagrange
equations are equivalent to a new system of differential equations
defined on the product of the configuration space with the Lie
algebra. These equations are called now in his honor
Euler-Poincar\'e equations. In modern language the contents of the
Poincar\'e's article \cite{Poincare} is presented for example in
\cite{Ho2011GM2,Marle}. English translation of the article
\cite{Poincare} can be found as Appendix D in \cite{Ho2011GM2}.

 \section{$G$-strand equations on matrix Lie algebras}

Denoting ${\sf m}:=\delta \ell/\delta{\sf u}$ and ${\sf n}:=\delta
\ell/\delta{\sf v}$ in $\mathfrak{g}^*$, the $G$-strand equations
become
$$
{\sf m}_t + {\sf n}_{s} - {\rm ad}^*_{\sf u}{\sf m}
 - {\rm ad}^*_{\sf v}{\sf n}
=0 \quad\hbox{and}\quad
\partial_t{\sf v} -\partial_{s}{\sf u} + {\rm ad}_{\sf u}{\sf v} =
0.
$$
For $G$ a semisimple \emph{matrix Lie group} and $\mathfrak{g}$ its
\emph{matrix Lie algebra} these equations become \begin{equation}
\label{MatrAlgEq} \begin{split} {\sf m}^{T}_t + {\sf n}^{T}_{s} +
{\rm ad}_{\sf u}{\sf m}^{T}
 + {\rm ad}_{\sf v}{\sf n}^{T}
=&0, \\
\partial_t{\sf v} -\partial_{s}{\sf u}  + {\rm ad}_{\sf u}{\sf v}=& 0
\end{split}
\end{equation} 
where the ad-invariant pairing for semisimple matrix Lie algebras is given by 

$$\Big\langle{{{\sf m}}}\,,\,{{{\sf n}}}\Big\rangle=\frac{1}{2}\tr({\sf m}^T{\sf n}), $$the transpose
gives the map between the algebra and its dual $(\,\cdot\,)^{T}:
\mathfrak{g}\to\mathfrak{g}^*$. For semisimple matrix Lie groups, the adjoint operator is the matrix 
commutator. Examples are studied in \cite{Ho-Iv-Pe, Ho-Iv2, FDT}.

\section{Lie-Poisson Hamiltonian formulation}

Legendre transformation of the Lagrangian $\ell({{{\sf u}},{{\sf
v}}}):\, \mathfrak{g}\times \mathfrak{g}\to\mathbb{R}$ yields  the
Hamiltonian $h({{{\sf m}},{{\sf v}}}):\, \mathfrak{g}^*\times
\mathfrak{g}\to\mathbb{R}$
\begin{equation}
h({{{\sf m}}},{{{\sf v}}}) = \Big\langle{{{\sf m}}}\,,\,{{{\sf
u}}}\Big\rangle - \ell({{{\sf u}},{{\sf v}}}) \,.
\label{leglagham} \vspace{-3mm}
\end{equation}

Its partial derivatives imply
\begin{eqnarray*}
\frac{\delta l}{\delta {{\sf u}}} = {{\sf m}} \,,\quad
\frac{\delta h}{\delta {{\sf m}}} = {{\sf u}} \quad\hbox{and}\quad
\frac{\delta h}{\delta {{\sf v}}} = -\,\frac{\delta \ell}{\delta
{{\sf v}}} = {\sf v} .
\end{eqnarray*}

These derivatives allow one to rewrite the Euler-Poincar\'e
equation solely in terms of momentum ${{\sf m}}$ as

\begin{equation} \label{hameqns-so3}
\begin{split}
{\partial_t} {{\sf m}} &= {\rm ad}^*_{\delta h/\delta {{\sf m}}}\,
{{\sf m}} + \partial_{s} \frac{\delta h}{\delta {{\sf v}}} - {\rm
ad}^*_{{\sf v}}\,\frac{\delta h}{\delta {{\sf v}}}
\, ,\\
\partial_t {{\sf v}}
&= \partial_{s}\frac{\delta h}{\delta {{\sf m}}} -  {\rm
ad}_{\delta h/\delta {{\sf m}}}\,{{\sf v}} \,.
\end{split}
\end{equation}
Assembling these equations into Lie-Poisson Hamiltonian form
gives,
%
%------------------------------------------
\begin{equation} \label{LP-Ham-struct-symbol1}
\frac{\partial}{\partial t}
    \begin{bmatrix}
    {{\sf m}}
    \\
    {{\sf v}}
    \end{bmatrix}
=
\begin{bmatrix}
  {\rm ad}^\ast(\,\cdot\,) {{\sf m}}
   &\hspace{5mm}
  \partial_s - {\rm ad}^*_{{\sf v}}
   \\
   \partial_s + {\rm ad}_{{\sf v}}
   &\hspace{5mm} 0
    \end{bmatrix}
    \begin{bmatrix}
   \delta h/\delta{{\sf m}} \\
   \delta h/\delta{{\sf v}}
    \end{bmatrix}
\end{equation}

%------------------------------------------

The Hamiltonian matrix in equation (\ref{LP-Ham-struct-symbol1})
also appears in the Lie-Poisson brackets for Yang-Mills plasmas,
for spin glasses and for perfect complex fluids, such as liquid
crystals.

\section{Example: The Euler-Poincar\'e PDEs for the $SO(3)$-strand  and the chiral model. The $2$-time spatial and body angular velocities on $\mathfrak{so}(3)$}

Let us make the following explicit identification:
\begin{equation}
{\sf u}=\left(%
\begin{array}{ccc}
  0 & -u_3 & u_2 \\
  u_3 & 0 & -u_1 \\
  -u_2 & u_1 & 0 \\
\end{array}%
\right)\in \mathfrak{g}\quad \leftrightarrow \quad \boldsymbol{{{\sf u}}}\equiv\left(%
\begin{array}{c}
  u_1 \\
  u_2 \\
  u_3 \\
\end{array}%
\right)\in \mathbb{R}^3 \label{eqomegaL} \end{equation}

and similarly for $\boldsymbol{{{\sf v}}}$. In terms of the
corresponding group element $g(s,t)$, describing rotation, $ {\sf
u}(t,{s})=g^{-1}\partial_t g(t,{s}) $ and $ {\sf
v}(t,{s})=g^{-1}\partial_{s} g(t,{s}) $ resemble $2$ body angular
velocities. For $G=SO(3)$ and Lagrangian $ \ell (\boldsymbol{{{\sf
u}},\,{{\sf v}}} ): \mathbb{R}^3\times\mathbb{R}^3\to\mathbb{R}, $
in $1+1$ space-time the Euler-Poincar\'e equation becomes

\begin{equation}
\frac{\partial}{\partial t} \frac{\delta\ell}{\delta \boldsymbol{{{\sf u}}}}
+ \boldsymbol{{{\sf u}}}\times\frac{\delta\ell}{\delta \boldsymbol{{{\sf u}}} }
=
- \left(\frac{\partial}{\partial {s}}  \frac{\delta\ell}{\delta \boldsymbol{{{\sf v}}} }
+
{\boldsymbol{{{\sf v}}}}\times
\frac{\delta\ell}{\delta \boldsymbol{{{\sf v}}} }\right)
\,,
\label{2timeEP-SO3}
\end{equation}

and its auxiliary equation becomes
\begin{equation}
\frac{\partial}{\partial t} \boldsymbol{{{\sf v}}}
=
\frac{\partial}{\partial {s}}{\boldsymbol{{{\sf u}}}}
+
{\boldsymbol{{{\sf v}}}}\times{\boldsymbol{{{\sf u}}}}
\,.
\label{aux-eqn-2timeX}
\end{equation}

%-------------------------Part f

The Hamiltonian form of these equations on $\mathfrak{so}(3)^*$ are obtained from the Legendre transform relations
\begin{eqnarray*}
\frac{\delta \ell}{\delta \boldsymbol{{{\sf u}}} }
= \boldsymbol{{{\sf m}}}
\,,\quad
\frac{\delta h}{\delta \boldsymbol{{{\sf m}}}} = \boldsymbol{{{\sf u}}}
\quad\hbox{and}\quad
\frac{\delta h}{\delta \boldsymbol{{{\sf v}}} }
= -\,\frac{\delta \ell}{\delta \boldsymbol{{{\sf v}}} }
\,.
\end{eqnarray*}

Hence, the  Euler-Poincar\'e equation implies the Lie-Poisson
Hamiltonian structure in vector form

\begin{equation*} \label{LP-Ham-struct-so3}
\partial_t
    \begin{bmatrix}
    \boldsymbol{{{\sf m}}}
    \\
    \boldsymbol{{{\sf v}}}
    \end{bmatrix}
=
\begin{bmatrix}
    \boldsymbol{{{\sf m}}}\times
   &     \partial_s
   + \boldsymbol{{{\sf v}}}\times
      \\
    \partial_s
   + \boldsymbol{{{\sf v}}}\times
& 0
    \end{bmatrix}
    \begin{bmatrix}
   \delta h/\delta\boldsymbol{{{\sf m}}} \\
   \delta h/\delta\boldsymbol{{{\sf v}}}
    \end{bmatrix}.
\end{equation*}

This Poisson structure appears in various other theories, such as
complex fluids and filament dynamics.

When \begin{equation} \label{spinchainLagr}\ell=\frac12 \int
(\boldsymbol{{{\sf u}}}\cdot A \boldsymbol{{{\sf u}}}
+\boldsymbol{{{\sf v}}}\cdot B \boldsymbol{{{\sf
v}}})\,ds\end{equation} this is the $SO(3)$ \emph{spin-chain
model}, which is in general non-integrable- eq.
(\ref{2timeEP-SO3}) and (\ref{aux-eqn-2timeX}) give:
\begin{equation} \frac{\partial}{\partial t}  A \boldsymbol{{{\sf
u}}} + \boldsymbol{{{\sf u}}}\times  A \boldsymbol{{{\sf u}}}+
\frac{\partial}{\partial {s}} B \boldsymbol{{{\sf v}}} +
{\boldsymbol{{{\sf v}}}}\times B \boldsymbol{{{\sf v}}} =0\,,
\label{SO3 spin chain1}
\end{equation}
\begin{equation}
\frac{\partial}{\partial t} \boldsymbol{{{\sf v}}} =
\frac{\partial}{\partial {s}}{\boldsymbol{{{\sf u}}}} +
{\boldsymbol{{{\sf v}}}}\times{\boldsymbol{{{\sf u}}}}
\,.\label{SO3 spin chain2}
\end{equation}

 When $A=-B=1$, this is the $SO(3)$ \emph{chiral model},
which is an integrable Hamiltonian system.

\begin{equation}
\boldsymbol{{{\sf u}}}_t - \boldsymbol{{{\sf v}}}_s =0\,,
\label{SO3 chiral1}
\end{equation}
\begin{equation}
\boldsymbol{{{\sf v}}}_t - {\boldsymbol{{{\sf u}}}}_s +
{\boldsymbol{{{\sf u}}}}\times{\boldsymbol{{{\sf v}}}}=0
\,.\label{SO3 chiral2}
\end{equation}

\section{Integrability}

Some of the $G$-strands models are well known integrable models.
They have a {\it zero-curvature representation} for two operators
$L$ and $M$ of the form
%\smallskip
\begin{align}
L_t - M_s + [L,M] = 0, \label{ZC-comrel}
\end{align}\bigskip
which is the compatibility condition for a pair of linear
equations
\[
\psi_s = L\psi, \quad\hbox{and}\quad \psi_t = M\psi.
\]

For the SO(3) chiral model for example these operators are
\begin{equation}
\begin{split} L&=\frac{1}{4}\left[(1+\lambda)({\sf u}-{\sf v})-
\left(1+\frac{1}{\lambda}\right)({\sf u}+{\sf v}) \right],\\
M&=-\frac{1}{4}\left[(1+\lambda)({\sf u}-{\sf v})+
\left(1+\frac{1}{\lambda}\right)({\sf u}+{\sf v}) \right].
\end{split}
\end{equation}

Another integrable matrix example: $SO(3)$ anisotropic Chiral
model \cite{Ch1981}
\begin{equation}\label{XY-eqn-1a}
\begin{split}
\partial_t{\boldsymbol{\mathsf{v}}}(t,{s}) - \partial_{s}{\boldsymbol{\mathsf{u}}}(t,{s})
+{\boldsymbol{\mathsf{u}}}\times
P{\boldsymbol{\mathsf{v}}}-{\boldsymbol{\mathsf{v}}}\times
P{\boldsymbol{\mathsf{u}}}=0 \,,
\\
\partial_{s}{\boldsymbol{\mathsf{v}}}(t,{s}) - \partial_t{\boldsymbol{\mathsf{u}}}(t,{s})
- {\boldsymbol{\mathsf{v}}}\times P
{\boldsymbol{\mathsf{v}}}+{\boldsymbol{\mathsf{u}}}\times
P{\boldsymbol{\mathsf{u}}}=0 \,.
%\label{XY-eqn-2a}
\end{split}
\end{equation}

$P=\text{diag}(P_1,P_2,P_3)$ is a constant diagonal matrix.  Under
the linear change of variables  \begin{equation}
\boldsymbol{\mathsf{X}} = \boldsymbol{\mathsf{u}} -
\boldsymbol{\mathsf{v}} \quad\hbox{and}\quad
\boldsymbol{\mathsf{Y}} = -\,\boldsymbol{\mathsf{u}} -
\boldsymbol{\mathsf{v}} \label{changeXY-uv}
\end{equation}
equations (\ref{XY-eqn-1a}) acquire the form of the following
$SO(3)$ anisotropic chiral model,

\begin{equation}\label{uv-eqn-1}
\begin{split}
\partial_t{\boldsymbol{\mathsf{X}}}(t,{s}) +\partial_{s}{\boldsymbol{\mathsf{X}}}(t,{s})
+{\boldsymbol{\mathsf{X}}}\times P{\boldsymbol{\mathsf{Y}}}&=0
\,,\\
\partial_t{\boldsymbol{\mathsf{Y}}}(t,{s}) - \partial_{s}{\boldsymbol{\mathsf{Y}}}(t,{s})
+{\boldsymbol{\mathsf{Y}}}\times P {\boldsymbol{\mathsf{X}}}&=0
\,.
%\label{uv-eqn-2}
\end{split}
\end{equation}

The system (\ref{uv-eqn-1}) represents two \emph{cross-coupled}
equations for ${\boldsymbol{\mathsf{X}}}$ and
${\boldsymbol{\mathsf{Y}}}$. These equations preserve the
magnitudes $|{\boldsymbol{\mathsf{X}}}|^2$ and
$|{\boldsymbol{\mathsf{Y}}}|^2$, so they allow the further
assumption that the vector fields
$(\boldsymbol{\mathsf{X}},\boldsymbol{\mathsf{Y}})$ take values on
the product of unit spheres $\mathbb{S}^2 \times \mathbb{S}^2
\subset  \mathbb{R}^3 \times\mathbb{R}^3$. The  anisotropic chiral
model is an integrable system and its Lax pair in terms of
$(\boldsymbol{\mathsf{u}},\boldsymbol{\mathsf{v}})$ utilizes the
following isomorphism between $\mathfrak{so}(3) \oplus
\mathfrak{so}(3)$ and $\mathfrak{so}(4)$:
\begin{equation}
A({\boldsymbol{\mathsf{u}}},{\boldsymbol{\mathsf{v}}})=\left(
\begin{matrix}
0 & u_3&-u_2 & v_1\\
-u_3 & 0&u_1 & v_2\\
u_2 & -u_1&0 & v_3 \\
-v_1 & -v_2&-v_3 & 0
\end{matrix}
\right) . \label{Mso4-def}
\end{equation}

The system (\ref{XY-eqn-1a}) can be recovered as a compatibility
condition of the operators
\begin{eqnarray}
L&=&\partial_{s}-A({\boldsymbol{\mathsf{v}}},{\boldsymbol{\mathsf{u}}})(\lambda\,{\rm Id}+J),\\
M&=&\partial_t-A({\boldsymbol{\mathsf{u}}},{\boldsymbol{\mathsf{v}}})(\lambda\,{\rm
Id}+J), \label{L-Mpair}
\end{eqnarray}
where the diagonal matrix $J$ is defined by
\begin{equation}
J = -\frac{1}{2}\text{diag}(P_1,P_2,P_3,P_1+P_2+P_3).
\label{J-def}
\end{equation}

This Lax pair is due to Bordag and Yanovski \cite{BoYa1995}. The
$O(3)$ anisotropic chiral model can be derived as an
Euler-Poincar\'e equation from a Lagrangian with quadratic kinetic
and potential energy. The details are presented in
\cite{Ho-Iv-Pe}.

\begin{remark}
If ${\sf P}={\rm Id}$, equations (\ref{XY-eqn-1a}) recover the
$SO(3)$ chiral model.
\end{remark}

\section{ The ${\rm Diff}(\mathbb{R})$-strand system}

The constructions described briefly in the previous sections can
be easily generalized in cases where the Lie group is the group of
the Diffeomorphisms. Consider Hamiltonian which is a
right-invariant bilinear form given by the $H^1$ Sobolev inner
product \begin{equation} H(u,v)\equiv \frac{1}{2}\int
_{\mathcal{M}}(uv+u_xv_x) dx.
\end{equation}
The manifold $\mathcal{M}$ is $\mathbb{S}^1$ or in the case when
the class of smooth functions vanishing rapidly at $\pm \infty$ is
considered, we will allow $\mathcal{M} \equiv \mathbb{R}$.

%(then again we can compactify $\mathcal{M}$ to $\mathbb{S}^1$ by
%identifying the functions at $+ \infty$ and $- \infty$).
Let us introduce the notation $u(g(x))\equiv u\circ g$. If
$g(x)\in G$, where $G\equiv \text{Diff}(\mathcal{M})$, then
$$H(u,v)=H(u\circ g, v\circ g)$$ is a right-invariant $H^1$ metric.

Let us further consider an one-parametric family of
diffeomprphisms, $g(x,t)$ by defining the $t$ - evolution as
\begin{equation}
\dot{g}=u(g(x,t),t), \qquad g(x,0)=x, \qquad \text{i.e.} \qquad
\dot{g}=u\circ g \in T_g G;
\end{equation} $u=\dot{g}\circ g^{-1}\in
\mathfrak{g}$, where $\mathfrak{g}$, the corresponding Lie-algebra
is the algebra of vector fields, $\text{Vect}(\mathcal{M})$. Now
we recall the following result:

\begin{theorem}(A. Kirillov, 1980, \cite{K81, K93})
The dual space of $\mathfrak{g}$ is a space of distributions but
the subspace of local functionals, called the regular dual
$\mathfrak{g}^*$  is naturally identified with the space of
quadratic differentials $m(x)dx^2$ on $\mathcal{M}$. The pairing
is given for any vector field $u\partial_x\in
\text{Vect}(\mathcal{M})$ by

$$\langle mdx^2, u\partial_x\rangle=\int_{\mathcal{M}}m(x)u(x)dx$$

\end{theorem}

The coadjoint action coincides with the action of a diffeomorphism
on the quadratic differential:

$$\text{Ad}_g^*:\quad mdx^2\mapsto m(g)g_x^2dx^2$$ and $$\text{ad}_{u}^*=2u_x+u\partial_x$$
Indeed, a simple computation shows that
\begin{eqnarray}
\langle\text{ad}_{u\partial_x}^*
mdx^2,v\partial_x\rangle&=&\langle
mdx^2,[u\partial_x,v\partial_x]\rangle=\int_{\mathcal{M}}m(u_xv-v_xu)dx=
\nonumber
\\ \int_{\mathcal{M}}v(2mu_x+um_x)dx&=&\langle(2mu_x+um_x)dx^2,v\partial_x\rangle, \nonumber \end{eqnarray}
i.e. $\text{ad}_{u}^*m=2u_xm+um_x$.

The ${\rm Diff}(\mathbb{R})$-strand system arises when we choose
$G={\rm Diff}(\mathbb{R})$. For a two-parametric group we have two
tangent vectors
\begin{equation*}
\partial_t{g}= u \circ g
\quad\hbox{and}\quad
\partial_s{g}= v \circ g
\,,
\end{equation*}
where the symbol $\circ$ denotes composition of functions.

In this right-invariant case, the $G$-strand PDE system with
reduced Lagrangian $\ell(u,v)$ takes the form,

\begin{align}
\begin{split}
\frac{\partial}{\partial t} \frac{\delta\ell}{\delta u} +
\frac{\partial}{\partial s}  \frac{\delta\ell}{\delta v } &= -\,
{\rm ad}^*_{u}\frac{\delta\ell}{\delta u} - {\rm ad}^*_{v
}\frac{\delta\ell}{\delta v } \,,
\\
\frac{\partial v}{\partial t}   - \frac{\partial u}{\partial s} &=
{\rm ad}_u v \,.
\end{split}
\label{Gstrand-eqn1R}
\end{align}\medskip

Of course, the distinction between the maps $({u},{v}):
\mathbb{R}\times \mathbb{R}\to \mathfrak{g}\times\mathfrak{g}$ and
their pointwise values $({u}(t,s),{v}(t,s))\in
\mathfrak{g}\times\mathfrak{g}$ is clear. Likewise, for the
variational derivatives ${\delta\ell}/{\delta {{u}}}$ and
${\delta\ell}/{\delta v}$.

\section{The ${\rm Diff}(\mathbb{R})$-strand Hamiltonian structure}

Upon setting  $m={\delta\ell}/{\delta u }$ and
$n={\delta\ell}/{\delta v }$, the right-invariant ${\rm
Diff}(\mathbb{R})$-strand equations in (\ref{Gstrand-eqn1R}) for
maps $\mathbb{R}\times\mathbb{R}\to G={\rm Diff}(\mathbb{R})$ in
one spatial dimension may be expressed as a system of two 1+2 PDEs
in $(t,s,x)$,
\begin{align}
\begin{split}
m_t + n_s &= -\, {\rm ad}^*_{u }m - {\rm ad}^*_{v }n = - (u m)_x -
m u_x - (v n)_x -  nv_x \,,
\\ \bigskip
v_t - u _s &= -\,{\rm ad}_v u = -u v_x + v u_x \,.
\end{split}
\label{Gstrand-eqn2R}
\end{align}
The Hamiltonian structure for these  ${\rm
Diff}(\mathbb{R})$-strand equations is obtained by Legendre
transforming to
\[
h(m,v)=\langle m,\, u\rangle - \ell(u,\,v) \,.\]

One may then write the equations (\ref{Gstrand-eqn2R}) in
Lie-Poisson Hamiltonian form as

\begin{equation}
\frac{d}{dt}
\begin{bmatrix}
m \\ v
\end{bmatrix}
=
\begin{bmatrix}
-\,{\rm ad}^*(\,\cdot\,) m  &\quad \partial_s + {\rm ad}^*_v
\\
\partial_s - {\rm ad}_v  &\quad 0
\end{bmatrix}
\begin{bmatrix}
{\delta h}/{\delta m} = u
\\
{\delta h}/{\delta v} = -\, n
\end{bmatrix}.
\label{1stHamForm}
\end{equation}

\section{Peakon solutions of the ${\rm Diff}(\mathbb{R})$-strand  equations}

With the following choice of Lagrangian,
 \begin{equation}
\ell (u ,v ) = \frac12 \|u \|^2_{H^1} - \frac12 \|v \|^2_{H^1} \,,
\label{Gstrand-pkn-Lag}
\end{equation}
the corresponding Hamiltonian is positive-definite and the ${\rm
Diff}(\mathbb{R})$-strand equations (\ref{Gstrand-eqn2R}) admit
peakon solutions in \emph{both} momenta
$$m=u-u_{xx} \quad \text{ and} \quad n=-(v-v_{xx}),$$ with
continuous velocities $u$ and $v$. This is a two-component
generalization of the CH equation.

\begin{theorem}
The ${\rm Diff}(\mathbb{R})$-strand equations (\ref{Gstrand-eqn2R})
admit singular solutions expressible as linear superpositions summed over $a\in\mathbb{Z}$
\begin{align}
\begin{split}
m(s,t,x) &= \sum_a M_a(s,t)\delta(x-Q^a(s,t)) \,,\\
n(s,t,x) &= \sum_a N_a(s,t)\delta(x-Q^a(s,t)) \,,
\\
u(s,t,x) & =K*m=\sum_a M_a(s,t) K(x,Q^a) \,,\\
v(s,t,x) & = -K*n=-\sum_a N_a(s,t) K(x,Q^a) \,,
\end{split}
\label{Gstrand-singsolns1}
\end{align}
that are \emph{peakons} in the case that $K(x,y)= \frac12
e^{-|x-y|}$ is the Green function the inverse Helmholtz operator
$(1-\partial_x^2)$: $$(1-\partial_x^2)K(x,0)=\delta(x) $$

\end{theorem}

The solution parameters $\{Q^a(s,t), M_a(s,t), N_a(s,t)\}$ with
$a\in\mathbb{Z}$ that specify the singular solutions
(\ref{Gstrand-singsolns1}) are determined by the following set of
evolutionary PDEs in $s$ and $t$, in which we denote $
K^{ab}:=K(Q^a,Q^b) $ with integer summation indices
$a,b,c,e\in\mathbb{Z}$:
\begin{align}
\begin{split}
\partial_t Q^a(s,t) &= u(Q^a,s,t) = \sum_b M_b(s,t) K^{ab}
\,,\\
\partial_s Q^a(s,t) &= v(Q^a,s,t) =-\sum_b N_b(s,t) K^{ab}
\,,\\
\partial_t M_a(s,t) &= -\, \partial_s N_a
-\sum_c (M_aM_c-N_aN_c) \frac{\partial K^{ac}}{\partial Q^a}
\quad\hbox{(no sum on $a$),}
\\
\partial_t N_a(s,t) &=-\partial_s M_a
+    \sum_{b,c,e} (N_bM_c - M_bN_c) \frac{\partial K^{ec}}{\partial Q^e} (K^{eb}-K^{cb})(K^{-1})_{ae}
\,.
\end{split}
\label{Gstrand-eqns}
\end{align}

The last pair of equations in (\ref{Gstrand-eqns}) may be solved as a system for the momenta, i.e., Lagrange multipliers $(M_a,N_a)$, then used in the previous pair to update the support set of positions $Q^a(t,s)$.

\section{Single-peakon solution of the of the ${\rm Diff}(\mathbb{R})$-strand system} The single-peakon solution of the ${\rm Diff}(\mathbb{R})$-strand equations (\ref{Gstrand-eqn2R}) is straightforward to obtain from (\ref{Gstrand-eqns}). Combining the equations in (\ref{Gstrand-eqns}) for a single peakon shows that $Q^1(s,t)$ satisfies the Laplace equation,
\begin{equation*}
(\partial_s^2 - \partial _t^2)Q^1(s,t)=0\,.
 \label{Q-1-peak}
\end{equation*}
Thus, any function $h(s,t)$ that solves the wave equation provides
a solution $Q^1=h(s,t)$. From the first two equations in
(\ref{Gstrand-eqns})
\begin{equation*}
 M_1(s,t)=\frac{1}{K_0}h_t (s,t)
 \qquad
 N_1(s,t)=\frac{1}{K_0}h_s (s,t),
 \label{MN-1-peak}
\end{equation*}
where $K_0=K(0,0)$.

%\begin{framed}
The solutions for the  single-peakon parameters $Q^1, M_1$ and
$N_1$ depend only on one function $h(s,t)$, which in turn depends
on the $(s,t)$ boundary conditions. The shape of the Green's
function comes into the corresponding solutions for the peakon
profiles
\[
u(s,t,x)  = M_1(s,t) K(x,Q^1(s,t)) \,,\qquad v(s,t,x)  =- N_1(s,t)
K(x,Q^1(s,t)) \,.
\]
%\end{framed}

\section{Peakon-Antipeakon collisions on a ${\rm Diff}(\mathbb{R})$-strand}

Denote the relative spacing $X(s,t)=Q^1-Q^2$ for the peakons at
positions $Q^1(t,s)$ and $Q^2(t,s)$ on the real line and the
Green's function $K=K(X)$.  Then the first two equations in
(\ref{Gstrand-eqns}) imply
\begin{align}
\begin{split}
\partial_t X &= (M_1-M_2)(K_0-K(X))
\,,\\
\partial_s X &= - (N_1-N_2)(K_0-K(X))
\,,\end{split} \label{Qdiff-eqns}
\end{align}
where $K_0=K(0)$.

The second pair of equations in (\ref{Gstrand-eqns}) may then be
written as
\begin{align}
\begin{split}
\partial_t M_1 &= - \partial_s N_1 - (M_1M_2- N_1N_2)K'(X)
\,,\\
\partial_t M_2 &= - \partial_s N_2 + (M_1M_2- N_1N_2)K'(X)
\,,\\
\partial_t N_1 &= - \partial_s M_1 + (N_1M_2-M_1N_2)
\frac{K_0-K}{K_0+K}K'(X)
\,,\\
\partial_t N_2 &= - \partial_s M_2 + (N_1M_2-M_1N_2)
\frac{K_0-K}{K_0+K} K'(X) \,.
\end{split}
\label{Gstrand-pp}
\end{align}
Asymptotically, when the peakons are far apart, the system
(\ref{Gstrand-pp}) simplifies, since $\frac{K_0-K}{K_0+K}\to1$ and
$K'(X)\to0$ as $|X|\to\infty$.

The system (\ref{Gstrand-pp}) has two immediate conservation laws
obtained from their sums and differences,
\begin{align}
\begin{split}
\partial_t (M_1+M_2) &= -\, \partial_s (N_1+N_2)
\,,\\
\partial_t (N_1-N_2) &=-\partial_s (M_1-M_2)
\,.\end{split} \label{pp-CLs}
\end{align}
These may be resolved by setting
\begin{align}
\begin{split}
M_1-M_2 &= \frac{\partial_t X}{K_0-K} \,,\qquad N_1-N_2 = -
\frac{\partial_s X}{K_0-K}
\,,\\
M_1+M_2 &= \partial_s\phi \,,\qquad N_1+N_2 = -\,\partial_t\phi
\,,
\end{split}
\label{pp-Xpotentials}
\end{align}
and introducing two potential functions, $X$ and $\phi$, for which
equality of cross derivatives will now produce the system of
equations (\ref{Qdiff-eqns}) and (\ref{Gstrand-pp}).

\section{A simplification.} A simplification arises if $\phi=0$, in
which case the collision is perfectly antisymmetric, as seen from
equation (\ref{pp-Xpotentials}). This is the peakon-antipeakon
collision, for which the  equation for $X$ reduces to
\begin{align}
(\partial_t^2 - \partial _s^2) X  &+ \frac{K'}{2(K_0-K)} (X_t^2-
X_s^2) = 0 \,.
\end{align}
This equation can be easily rearranged to produce a linear
equation:
\begin{align}
(\partial_t^2 - \partial _s^2) F(X)  = 0 \,,\quad\hbox{where}\quad
F(X) = \int_{X_0}^X (K_0-K(Y))^{-1/2}\,dY \,.
\end{align}
When $K(Y)=\frac{1}{2}e^{-|Y|}$, we have
\begin{align}
F(X) = \sqrt{2}\int_{X_0}^X \frac{1}{\sqrt{1-e^{-|Y|}}}\,dY
.\label{F-express}
\end{align}

We can take for simplicity $X_0=0$, this would change $F(X)$ only
by a constant. The computation gives $$F(X)=
2\sqrt{2}\,\text{sign}(X)\cosh^{-1}\left(e^{|X|/2}\right) $$.
Hence the solution $X(t,s)$ can be expressed in terms of any
solution $h(t,s)$ of the linear wave equation $(\partial_t^2 -
\partial _s^2)h(t,s)=0$ as
\begin{align}
X(t,s) = \pm \ln \left({\,\rm cosh}^2(h(t,s))\right) \,.
\label{X-express}
\end{align}
$h(t,s)$ is any solution of the wave equation. $$M_1=-M_2 =
\frac{\partial_t X}{2(K_0-K(X))} \,,\qquad N_1=-N_2 = -
\frac{\partial_s X}{2(K_0-K(X))}.$$

\section*{Complex ${\rm Diff}(\mathbb{R})$-strand  equations }

The ${\rm Diff}(\mathbb{R})$-strands may also be
\emph{complexified}. Upon complexifying $(s,t)\in
\mathbb{R}^2\to(z,\bar{z})\in \mathbb{C}$ where $\bar{z}$ denotes
the complex conjugate of $z$ and setting $\partial_z{g}= u \circ
g$ and $\partial_{\bar{z}}{g}= \bar{u} \circ g$ the
Euler-Poincar\'e $G$-strand equations in (\ref{Gstrand-eqn2R})
become
\begin{align}
\begin{split}
\frac{\partial}{\partial z} \frac{\delta\ell}{\delta u} +
\frac{\partial}{\partial \bar{z} }  \frac{\delta\ell}{\delta {
\bar{u}} } &= -\, {\rm ad}^*_{u}\frac{\delta\ell}{\delta u} - {\rm
ad}^*_{{\bar{u}} }\frac{\delta\ell}{\delta { \bar{u}} } \,,
\\
\frac{\partial {\bar{u}}}{\partial z}   - \frac{\partial
u}{\partial \bar{z}} &=  {\rm ad}_u { \bar{u}} \,.
\end{split}
\label{ComplexGstrand-eqns1}
\end{align}
Here the Lagrangian $\ell$ is taken to be real: \begin{align} \ell
(u, {\bar{u}}) = \frac12\|\nu\|_{H^1}^2 = \frac12\int u\,
(1-\partial_x^2) \, {{ \bar{u}}} \,dx . \, \label{ComplexNorm-ell}
\end{align}

%\section*{\huge\color{red}Complex ${\rm Diff}(\mathbb{R})$-strand  equations II}

Upon setting  $m={\delta\ell}/{\delta u}$,
$\bar{m}={\delta\ell}/{\delta {\bar{u}} }$, for the real
Lagrangian $\ell$, equations (\ref{ComplexGstrand-eqns1}) may be
rewritten as

\begin{align}
\begin{split}
m_{z} +  \bar{m}_{\bar{z}} &= -\, {\rm ad}^*_{u }m - {\rm
ad}^*_{{\bar{u}} }\,\bar{m} = - (u m)_x - mu_x - ({\bar{u}}
\,\bar{m})_x - \bar{m}\,{\bar{u}}_x \,,
\\  \\
{\bar{u}}_z- u _{\bar{z}} &= -\,{\rm ad}_{{\bar{u}}} \,u = -u \,{
\bar{u}}_x + { \bar{u}}\,u_x \,,
\end{split}
\label{ComplexGstrand-eqns2}
\end{align}\bigskip

where the independent coordinate $x\in\mathbb{R}$ is on the real
line, although coordinates $ (z,\bar{z})\in\mathbb{C}$ are
complex, as are solutions $u$, and $m=u-u_{xx}$. This is a
possible comlexification of the Camassa-Holm equation. These
equations are invariant under two involutions, $P$ and $C$, where
\[
P: (x,m)\to(-x,-m)
\quad \hbox{and} \quad
C: \hbox{Complex conjugation.}
\]
They admit singular solutions just as before, modulo
$\mathbb{R}\times\mathbb{R}\to\mathbb{C}$. For real variables
$m=\bar{m}$, $u={\bar{u}}$ and real evolution parameter
$z=\bar{z}=:t$, they reduce to the CH equation. Their travelling
wave solutions and other possible CH complexifications are studied
in \cite{Ho-Iv1}.

\section*{Conclusions}\label{conclusion-sec}

The $G$-strand equations comprise a system of PDEs obtained from
the Euler-Poincar\'e (EP) variational equations for a
$G$-invariant Lagrangian, coupled to an auxiliary
\emph{zero-curvature} equation. Once the $G$-invariant Lagrangian
has been specified, the  system of $G$-strand equations in
(\ref{MatrAlgEq}) follows automatically in the EP framework. For
matrix Lie groups, some of the the $G$-strand systems are
integrable. The single-peakon and the peakon-antipeakon solution
of the ${\rm Diff}(\mathbb{R})$-strand equations
(\ref{Gstrand-eqn2R}) depends on a single function of $s,t$. The
\emph{complex} ${\rm Diff}(\mathbb{R})$-strand equations and their
peakon collision solutions have also been solved by elementary
means. The stability of the single-peakon solution under
perturbations into the full solution space of equations
(\ref{Gstrand-eqn2R}) would be an interesting problem for future
work.

\section{Acknowledgments}
We are grateful for enlightening discussions of this material with F. Gay-Balmaz, T. S. Ratiu and C. Tronci. Work by RII was partially supported by the Science Foundation Ireland (SFI), under Grant No. 09/RFP/MTH2144. Work by DDH was partially supported by Advanced Grant 267382 FCCA from the European Research Council.

\end{document}